\documentclass[a4paper]{cas-dc}

\usepackage[numbers]{natbib}
\usepackage{amsthm,amsfonts,amssymb}
\usepackage{color,abraces,tcolorbox}
\usepackage{mathtools}
\usepackage{times}

\def\I{{}^{\tt I}}
\def\B{{}^{\tt B}}
\def\Exp{\mbox{Exp}}

\def\calx{\mathcal{X}}
\def\calj{\mathcal{J}}
\def\calh{\mathcal{H}}
\def\rea{\mathbb{R}}
\def\bfx{\mathbf{x}}
\def\bfp{\mathbf{p}}
\def\bfw{\mathbf{w}}

\newtheorem{remark}{{\it Remark}\rm}

\newtheorem{definition}{Definition}

\newtheorem{assumption}{Assumption}

\newtheorem{proposition}{Proposition}

\def\qed{\hfill $\square$}

\def\col{\mbox{col}}

\def\begmat#1{\begin{bmatrix}#1\end{bmatrix}}

\usepackage{color}

\begin{document}
\let\WriteBookmarks\relax
\def\floatpagepagefraction{1}
\def\textpagefraction{.001}

\shorttitle{}

\shortauthors{Yi and Manchester}

\title [mode = title]{On IMU preintegration: A nonlinear observer viewpoint and its applications}                      


\author[1]{Bowen Yi}[
                        auid=000,bioid=1,
                        ]

\cormark[1]


\ead{b.yi@outlook.com}


\affiliation[1]{organization={Robotics Institutes, University of Technology Sydney},
    postcode={NSW 2006}, 
    country={Australia}}

\affiliation[2]{organization={Australian Centre for Robotics, School of Aerospace, Mechanical and Mechatronic Engineering, The University of Sydney},
    postcode={NSW 2006}, 
    country={Australia}}

\author[2]{Ian R. Manchester}

\ead{ian.manchester@sydney.edu.au}

\cortext[cor1]{Corresponding author (The work has partially been done when the first author was with The University of Sydney.)}



\begin{abstract}
The inertial measurement unit (IMU) preintegration approach nowadays is widely used in various robotic applications. In this article, we revisit the preintegration theory and propose a novel interpretation to understand it from a nonlinear observer perspective, specifically the parameter estimation-based observer (PEBO). We demonstrate that the preintegration approach can be viewed as recursive implementation of PEBO in moving horizons, and that the two approaches are equivalent in the case of perfect measurements. We then discuss how these findings can be used to tackle practical challenges in estimation problems. As byproducts, our results lead to a novel hybrid sampled-data observer design and an approach to address statistical optimality for PEBO in presence of noise.
\end{abstract}


\maketitle

\section{Introduction}

State estimation and perception are fundamentally important for autonomous systems \cite{BAR}. Initially, filtering approaches dominated the field of \emph{online} state estimation due to the limitation of computational capacity \cite{BERetal,DISetal}. In recent years full smoothing approaches which are based on nonlinear batch optimisation have gained popularity in numerous localisation problems, since they provide estimates with high accuracy \cite{STRetal}. However, the optimisation-based estimation framework is computationally demanding. This issue is currently becoming more urgent than ever as we have witnessed the trend of utilisation of monocular cameras with IMUs -- known as the monocular visual-inertial system (VINS) -- in real-world robotic systems. The VINS is an asynchronous sampled system, with IMUs providing measurements at a high rate. As a result, there is the need to calculate the ``standard'' inertial integration from initial conditions between two camera frames, which thus makes it a daunting task to solve in real time.

In \cite{LUPSUK}, Lupton and Sukkarieh propose the IMU preintegration approach to address the above-mentioned computational challenges. It allows pre-processing of the high-rate data from IMU to obtain low-rate pseudo measurements, in which initial conditions and the preintegrated quantities are separated, thus reducing on-line computational burden significantly. Later on, the preintegration approach was extended to kinematic models living on nonlinear manifolds \cite{FORetal}, and now is gradually becoming a popular result in the robotics community. More recently, it has been improved and elaborated from several different perspectives, e.g., analytical solutions for graph optimisation \cite{ECKetal}, approximation via Gaussian process \cite{LEGetal}, and generalisation on groups \cite{BARBON}, just to name a few. Since its introduction, the preintegration approach has been widely applied in various robotic systems, see e.g. \cite{BROetal,FOUetal,QINetal}.

In this paper we prove that the preintegration approach can be derived following the observer theory for nonlinear systems, in particular the parameter estimation-based observer (PEBO). It is a novel kind of constructive observer technique recently proposed by Ortega \emph{et al.} in \cite{ORTetalscl} and later elaborated in \cite{ORTetalAUT}, in which state observation is translated into an on-line parameter identification problem; see \cite{YIetalTAC18} for a geometric interpretation. Recently, we have extended the PEBO methodology from Euclidean space to marix Lie groups, which has been proven instrumental in solving several open problems in observer design for robotic systems \cite{YIetalAUT,YIetalCDC,YIWAN}.

Although the approaches of preintegration and PEBO have been pursued in parallel in different communities, it is interesting and generally important to elucidate the connections between these two frameworks. By bridging these distinct bodies of research, this paper aims to unveil their relationship and present the following main contributions.

\begin{itemize}
    \item[1)] We revisit the preintegration theory and provide a nonlinear observer interpretation to it. Namely, the preintegrated signals are exactly the dynamic extended variables (i.e., fundamental matrices) in PEBO but implemented in \emph{a moving horizon}. Under some mild assumptions, we establish the equivalence between the preintegration and PEBO approaches.
    
    \item[2)] We show the practical utility of the resulting equivalence in addressing several practical challenges encountered in state estimation problems. In particular, it provides a novel solution to design sampled-data observers for continuous-time dynamical systems and enables the attainment of statistical optimality in PEBO in the presence of noisy measurements.
    
\end{itemize}

The remainder of the paper is organised as follows. In Section \ref{sec2} we consider the dynamical models in Euclidean space as an illustrative example to recall the preintegration and PEBO approaches. It is followed by some preliminary results about the connections between two approaches in Euclidean space. In Section \ref{sec3}, we present our main results on the manifold $SO(3) \times \rea^n$, which is the state space considered in numerous robotic and navigation-related problems, and also the original motivation of IMU preintegration. Then, we discuss some applications of the main claim in Section \ref{sec4}. The paper is wrapped up by some concluding remarks in Section \ref{sec5}.

\emph{Notation:} For a given variable or signal $x$, sometimes we may simply write $x(t)$ as $x_t$, and the dependency of signals on $t$ is omitted for brevity when clear. We use $x(t_1^-)$ to denote the value of $x$ just before $t_1$, i.e. $x(t_1^-):= \lim_{s>0,s\to 0} x(t_1 - s) $. We use $|x|$ to represent the standard Euclidean norm of a vector. $SO(3)$ represents the special orthogonal group, which is defined as $SO(3)=\{R\in \rea^{3\times3}|R^\top R = I_3, ~ \det(R) =1\}$. The operator $(\cdot)_\times$ is defined such that $a_\times b := a\times b$ for two vectors $a, b \in \rea^n$. For a variable $y$, we use $\bar y$ to represent its noisy measurement from sensors. $\lambda_{\tt max}\{A\}$ denotes the largest eigenvalue of a symmetric matrix $A\in \rea^{n\times n}$.

\section{Preliminary Results in Euclidean Space}
\label{sec2}

We start with the deterministic systems with states living in Euclidean space to introduce our basic idea. Its extension to the systems on manifolds, which is tailored for pose estimation of rigid bodies, will be presented in the next section.

\subsection{Problem Set}

In many engineering problems, there is a need to estimate the unknown internal state $x \in \mathcal{X} \subset \rea^n $ for the linear time-varying (LTV) dynamical system
\begin{equation}
\label{LTV:1}
\begin{aligned}
  \dot x & ~=~ A_t x + B_t u
  \\
  y & ~=~ C_t x + D_t u
\end{aligned}
\end{equation}
with input $u \in \rea^m$ and the output $y \in \rea^p$, and we usually consider the state space $\calx$ as $\rea^n$. Since sensor noise is unavoidable in practice, the measured signals of $u$ and $y$ satisfy
\begin{equation}
\label{noisy:uy}
\bar u = u + \epsilon_u, \quad \bar y = y + \epsilon_y,
\end{equation}
in which $\epsilon_u \in \rea^m $ and $ \epsilon_y\in \rea^p$ represent measurement noise, usually modelled as zero-mean white-noise processes. This estimation problem has been well addressed by the Kalman filter and the full-information estimation approach (a.k.a. batch optimisation). 


In some applications, despite admitting continuous-time models, we are concerned with estimation of the state $x$ at some discrete instants $\{t_k\}_{k \in \mathbb{N}}$. This is because multiple sensors provide information with different rates -- sometimes even having obvious time-scale separation. For example, in the problem of visual inertial navigation (VIN) for robotics, the IMU provides data at a very high rate, and it is reasonable to roughly view inertial measurements as some continuous-time signals. In contrast, it is well-known that image processing is relatively computationally heavy, and thus the camera provides data at a low rate. As a result, if the estimation algorithm is being processed at the same rate as the IMU, then it is usually not tractable on-line. 

In this paper, we make the following assumption. This scenario exists in many practical problems, particularly in robotic systems.

\begin{assumption}\rm\label{ass:1}
The input $\bar u$ is available as a continuous-time signal, and the output $\bar y$ is measured at some discrete instances $\{t_k\}_{k\in \mathbb{N}}$.
\end{assumption}

The main results can be extended straightforwardly to discrete-time systems with multi-rate sampled data (i.e. high-frequency input $\bar u$ and low-rate output $\bar y$), and we do not discuss it in this paper.

\subsection{Preintegration in Euclidean Space}

To address the state estimation of $x$ under Assumption \ref{ass:1}, Lupton and Sukkarieh proposed in \cite{LUPSUK} the preintegration approach to generate pseudo-measurements to improve on-line efficiency. Let us recall its basic idea with the LTV model \eqref{LTV:1} as follows.

\begin{proposition}\label{prop:1}
\rm \cite{BARBON} Consider the LTV system \eqref{LTV:1}. Given two instants $t_k< t_{k+1}$, there exist a matrix $F_k$ and a vector $v_k$ such that the state satisfies
\begin{equation}
\label{xk2k+1}
   x(t_{k+1}) = F_k  x(t_k) + v_k.
\end{equation}
for all $x({t_k}) \in \rea^n$. \qed
\end{proposition}

Its proof is available in \cite{BARBON}. We underscore that the matrices $F_k$ and $v_k$ are independent of the state $x$, which are accessible signals and uniquely determined by the measurable signals $A_t,B_t$ and $u_t$. Hence, we call $F_k$ and $v_k$ as \emph{preintegration}, and they can be calculated as
$$
F_k = F(t_{k+1}^-), \quad v_k = v(t_{k+1}^-),
$$
which is generated by the dynamics
\begin{equation}
\label{pre_int}
\left.
\begin{aligned}
   \dot F & = A_t F, \quad  F(t_k^+) = I_n
   \\
   \dot v & = A_t v + B_tu, \quad  v(t_k^+) = 0_n
\end{aligned}
~\right\}
\quad \mbox{Preintegration}
\end{equation}
Note that when implementing preintegration we only have the measurable signal $\bar y$ rather than the perfect output $y$, and thus the second preintegration is implemented as
\begin{equation}
\label{pre_int2}
\dot{\bar v} = A_t \bar v + B_t\bar u , \quad \bar v(t_k^+) =0_n,
\end{equation}
where $\bar v$ may be viewed as the noisy signal of $v$. They can be written as the Picard integral for $t \in ( t_k, t_{k+1} )$
$$
\begin{aligned}
F_t  =  \int_{t_k}^t  A_s F_s ds, \quad
\bar v_t  =  \int_{t_k}^t \Big( A_s v_s + B_s \bar u_s \Big) ds,
\end{aligned}
$$
and implemented numerically via discretization.

Now, using the preintegration we obtain the equation \eqref{xk2k+1} that is a new LTV discrete-time dynamical model with known $F_k$ and $v_k$, and the (nominal) output function
\begin{equation}
\label{ytk}
y(t_k) = C(t_k) x(t_k) + D(t_k) u(t_k).    
\end{equation}
Supposed the current moment is $t_N$, in the full-information estimation (FIE) approach we need to estimate $\{x(t_k)\}_{k\in \ell}$ with $\ell:= \{0,\ldots, N\}$. The simplest case is to consider solving the optimisation 
\begin{equation}
\label{opt:x}
\begin{aligned}
&(\hat \bfx, \hat \bfw)  = \underset{\bfx, \bfw}{\arg\min} ~ J_{\tt x}(\bfx) + J_{\tt w}(\bfw)
\\
&{\rm s.t.} \quad  \hat x(t_{k+1}) - F_k  \hat x(t_k) - \bar{v}_k = w_k
\end{aligned}
\end{equation}
with the cost functions\footnote{We assume that $y$ is measured from $t_0$ without loss of generality.}
\begin{equation}
\label{J}
\begin{aligned}
  J_{\tt x}(\bfx) & = \sum_{k = 0}^{N-1} \gamma_k \big|\bar y(t_k) - C(t_k) x(t_k) - D(t_k) \bar{u}(t_k) \big|^2
  \\
 J_{\tt w}(\bfw) & = \sum_{k = 0}^{N-1} \gamma_k'|w_k|^2
\end{aligned}
\end{equation}
and the definitions
$$
\begin{aligned}
\bfx & := \col( x_0, \ldots,  x_{N}), \quad \hat\bfx := \col( \hat x_0, \ldots,  \hat x_{N})
\\
\bfw & := \col(w_0,\ldots, w_N).
\end{aligned}
$$
The coefficients $\gamma_k, \gamma_k'>0$ may be involved to weight different instances, and two widely-used selections are:
\begin{itemize}
    \item[(i)] using the norm inverse of some covariance for the consideration of noise; and 
    \item[(ii)] selecting $\gamma_k = \lambda^{N-k}$ with $\lambda \in (0,1)$ to represent forgetting factors in on-line deterministic estimators.
\end{itemize}

The above summary of preintegration is presented as a high-level framework, which may be implemented in different ways. For instance, the optimisation problem can be solved for each instance (a.k.a. full-information estimation, FIE), in a moving-horizon, or incrementally as done in LTV Kalman-Bucy filters at the discrete instants $\{t_k\}$ in a lower sampling rate; the optimisation may also be replaced by computing the optimal maximum a posteriori (MAP) estimate, and combined with factor graphs.

To summarise, the basic idea is to use the preintegration \eqref{xk2k+1} to transform the continuous-time model \eqref{LTV:1} into the discrete model \eqref{xk2k+1} with low-rate measurements, and then complete the estimation task.\footnote{To distinguish from the other estimates in the remainder of the paper, we write the estimate from the preintegration approach as $\hat \bfx_{\tt PI}$.} Note that a salient feature of \eqref{xk2k+1} is the separation between the preintegrated signals $F,v$ and the initial condition $x(t_k)$, which is capable of reducing significant on-line computational burden in the nonlinear context.

\vspace{0.2cm}

\fbox{ \parbox { .9\linewidth} 
{
{\em State Estimation via Preintegration:} 
\begin{itemize}
    \item[-] preintegration: \eqref{pre_int}
    \item[-] estimate: $\hat \bfx_{\tt PI}$
    \item[-] optimisation: \eqref{opt:x}-\eqref{J}
\end{itemize}
}}
\vspace{0.2cm}

\begin{remark}
\rm
The computational burden of estimation of the original continuous-time system \eqref{LTV:1} is not prohibitive, due to \emph{linearity} in the model. However, when considering the visual navigation problem on manifolds, high nonlinearity and non-convexity  limit the performance and complicate the analysis of both full-information estimation and filtering approaches.
\end{remark}


\subsection{Parameter Estimation-Based Observer}

Recently, a new constructive nonlinear observer technique, namely PEBO, has been developed for a class of state-affine systems \cite{ORTetalAUT,ORTetalscl}. Its basic idea is translating state estimation into the one of some constant variables and then identifying them online. This provides an efficient way to simplify observer design.

Instead of introducing the approach comprehensively, we limit ourselves to the LTV system \eqref{LTV:1} to show the basic idea of PEBO. Following \cite{ORTetalAUT}, the first step is to design the dynamic extension 
\begin{equation}
\label{dyn_ext}
\dot \xi = A_t \xi + B_t u, \quad \xi(t_0)= \xi_0,
\end{equation}
with $\xi\in \rea^n$, in which the initial condition $\xi_0$ is selected by users thus being known. We underline here that the PEBO approach is developed for the deterministic system with the perfect measurement $u$, and the robustness to various uncertainties can be addressed from standard Lyapunov analysis. In this subsection, we consider the case with access to the perfect $u$, and its extension to with the noisy measurement $\bar u$ will be discussed in Section \ref{sec4}.

If we define the error $e:= x -\xi$, it yields the error dynamics
$$
\dot e = A_t e.
$$
As shown in linear systems theory \cite{RUG}, the solution of $e$ is given by $e(t) = \Phi(t,0) e(0)$, in which $\Phi(t,s)$ is the state transition matrix of $A_t$ from $s$ to $t$. Though it is generally impossible to write down the function $\Phi(t,s)$ analytically, it can be calculated by implementing the dynamics of fundamental matrix $\Omega$ on-line
\begin{equation}
\label{dot:omega}
\begin{aligned}
\dot \Omega & ~=~  A_t \Omega, \qquad~~~~ \Omega(t_0)= I_n
\\
\Phi(t,s) & ~=~  \Omega(t) \Omega(s)^{-1}.
\end{aligned}
\end{equation}
Then, we have the new parameterisation to the state $x$ as
$
 x_t = \xi_t - \Omega_t \xi_0 + \Omega_t \theta
$
with the unknown vector $\theta := x(0)$. It means that once the parameter $\theta$ have been determined as $\hat\theta$, one has the state estimation as 
\begin{equation}
\label{hat:x}
\hat x_t = \xi_t - \Omega_t \xi_0 + \Omega_t \hat\theta.
\end{equation}
By plugging the new parameterization of $x$ into \eqref{LTV:1}, we have the linear regression model with respect to $\theta$ as follows
\begin{equation}
\label{LRE:1}
 Y_t = C_t\Omega_t \theta
\end{equation}
with the variable
$
 Y_t := y_t - C_t \xi_t + C_t\Omega_t \xi_0 - D_t \bar u_t.
$
Its noisy ``measurement'' is defined accordingly as
\begin{equation}
\label{Y}
\bar Y_t := \bar y_t - C_t \xi_t + C_t\Omega_t \xi_0 - D_t \bar u_t.
\end{equation}
The remainder is to estimate $\theta$ from the regressor \eqref{LRE:1} on-line. With measurements collected at $\{t_k\}_{k \in \mathbb{N}}$, the simplest case at the moment $t_N$ is to solve the optimisation
\begin{equation}
\label{hat:theta}
\hat\theta:= \underset{\theta \in \rea^n}{\arg\min} ~ \sum_{k=0}^{N-1} \gamma_k \Big|  \bar Y(t_k) - C(t_k)\Omega(t_k)\theta \Big|^2,
\end{equation}
with some coefficients $\gamma_k >0$.

Hence, the PEBO approach can be summarised below.

\vspace{0.2cm}

\fbox{ \parbox { .9\linewidth} 
{
{\em Parameter Estimation-Based Observer:} 
\begin{itemize}
    \item[-] dynamics: \eqref{dyn_ext}, \eqref{dot:omega}
    \item[-] estimate (observer output): $\hat \bfx_{\tt PEBO}$ from \eqref{hat:x}
    \item[-] optimisation: \eqref{hat:theta}
\end{itemize}
}}
\vspace{0.2cm}

\begin{remark}
\rm 
For batch optimisation or filtering approaches, it is necessary to impose some ``informative'' excitation or observability assumptions on the model \eqref{LTV:1} along the trajectory. There are some observer design tools requiring observability/detectability \emph{uniformly along all feasible solutions}, e.g., \cite{BERetal,BES,KHA,YIetalTAC22}; however, this is not the case in various robotic localisation and navigation problems. It means that the optimisation \eqref{hat:theta} may have multiple or even infinite solutions under an insufficiently excited trajectory.
\end{remark}

\subsection{The PEBO Viewpoint to Preintegration}

In this section, we provide our new interpretation to the preintegration approach from a nonlinear observer perspective. For states living in Euclidean space with perfect or noise-free measurement of $u$, we summarise our findings as follows.

\begin{proposition}
\label{prop:eul}\rm 
Consider the LTV system \eqref{LTV:1} with $\epsilon_u =0$. State estimation from the  preintegration approach using \eqref{pre_int}-\eqref{J} exactly coincides with that from the PEBO \eqref{dyn_ext}-\eqref{hat:theta} using the zero initial condition $\xi_0 = 0_n$, in the following senses.
\begin{itemize}
    \item[a)] The preintegration signal $F$ and the fundamental matrix $\Omega$ satisfy
    \begin{equation}
    \label{itema}
    \begin{aligned}
        \Omega(t_k) & ~=~ \prod_{i=0}^{k-1} F_i :=  F_{k-1} \ldots F_0 , \quad \forall k\in \mathbb{N}
        \\
        \Omega(t) & ~=~ F(t)\Omega(t_k) , \quad t \in (t_k, t_{k+1}).
    \end{aligned}
    \end{equation}
    
    \item[b)] The preintegration signal $v$ and the dynamic extension variable $\xi$ verify
    \begin{equation}
        v_t ~=~ \xi_t - \Omega_t \Omega(t_k)^{-1} \xi(t_k)  , \quad t \in (t_k, t_{k+1}).
    \end{equation}
    
    \item[c)] If the cost function $J_{\tt x} + J_{\tt w}$ in \eqref{J} admits a unique global minimum, the PEBO estimate equals to the one from preintegration, i.e., $\hat \bfx_{\tt PEBO} = \hat \bfx_{\tt PI}$.  
\end{itemize}
\end{proposition}
\begin{proof}
First, we note that the fundamental matrix $\Omega$ shares the same dynamics as the one of the matrix $F$ in preintegration. The only difference is that the latter resets its initial values in instances $\{t_k^+\}_{k \in \mathbb{N}}$. From the semigroup property of the state transition matrix $\Phi(t,s)$, as well as the resetting 
$
\lim_{s\to t_k^{-}}F(s) = I_n,
$
we have
$$
\Phi(t,t_k) = F(t), \quad t\in (t_k, t_{k+1}).
$$
On the other hand, for $t\in (t_k, t_{k+1})$ we have
$$
\begin{aligned}
 \Omega(t) &~=~ \Phi(t,t_0) \Omega_0
 \\
 &~=~ \Phi(t,t_k)\Phi(t_{k},t_{k-1}) \cdots \Phi(t_1,t_0) I_n
 \\
 & ~=~ F(t) \prod_{i=0}^{k-1} F_i,
\end{aligned}
$$
which verifies the first claim.

For the case $\epsilon_u =0$, we have $v(t) = \bar v(t)$ for all $t\ge 0$. By comparing the dynamics of $\xi$ and $v$, we have
$$
\dot{\aoverbrace[L1R]{v-\xi}} ~=~ A_t (v-\xi),
$$
thus
$$
v_t - \xi_t = \Phi(t,s) (v_s -\xi_s), \quad \forall  t_{k+1}>t\ge s > t_k.
$$
Selecting $s= t_k$, and resetting as done in preintegration
$
\lim_{s\to t_k^{-}}  v(s) = 0_n,
$
then for $t\in (t_k,t_{k+1})$ we have
$$
v_t ~=~ \xi_t - \Phi(t,t_k)\xi(t_k),
$$
which verifies the item b).

At the end, let us show the equivalence between the optimisation problems \eqref{opt:x} and \eqref{hat:theta}. For the case of $\epsilon_u = 0$ (with perfect measurement of $u$), we have $ x(t_{k+1}) - F_k  x(t_k) - v_k =0$, and thus
\begin{equation}
\label{JJ}
J_{\tt x}(\bfx)  = J_{\tt x}(\bfx) + J_{\tt w}(0) \le J_{\tt x}(\hat \bfx ) + J_{\tt w}(\hat \bfw).
\end{equation}
Since we have assumed the unique minimum of the cost function, the optimisation in the preintegration approach becomes
$$
\hat \bfx  = \underset{\bfx}{\arg\min} ~ J_{\tt x}(\bfx) 
$$
with the hard constraint
\begin{equation}
\label{cs:1}
\hat x(t_{k+1}) - F_k  \hat x(t_k) -v_k =0, \quad  k =0,\ldots, N-1.
\end{equation}

Invoking the properties a)-b), the above optimisation can be written as
\begin{equation}
\label{opt:2}
\begin{aligned}
\hat \bfx & ~=~ \underset{\bfx \in \rea^{(N+1)n}}{\arg\min} ~ \sum_{k =0}^{N-1} \gamma_k \big|\bar y(t_k) - C(t_k) x(t_k) - D(t_k) \bar u(t_k) \big|^2
\\
& ~=~ \underset{\bfx \in \rea^{(N+1)n}}{\arg\min} ~ \sum_{k =0}^{N-1}
\gamma_k \left|\bar y_{t_k} - C_{t_k} \big(\xi_{t_k} +\Omega_{t_k}( x_0 - \xi_0)\big) \right.\\
& \left. \hspace{6.6cm} - D_{t_k} u_{t_k} \right|^2
\\
 & ~=~ \underset{\bfx \in \rea^{(N+1)n}}{\arg\min} ~ \sum_{k =0}^{N-1}
 \gamma_k \Big| \bar Y(t_k) - C(t_k)\Omega(t_k) x_0 \Big|^2,
\end{aligned}
\end{equation}
where in the last equation we have used the hard constraints \eqref{cs:1}. Let us recursively solve  \eqref{cs:1} -- combining the properties a) and b) -- we have the new constraint 
\begin{equation}
\label{cs:2}
\hat x(t_k) = \xi(t_k) + \Omega(t_k) \hat x_0,
\end{equation}
which has been plugged into the second equation in \eqref{opt:2}.

It is clear that the cost function in \eqref{opt:2} only contains the decision variable $\hat x_0$, which is the first $n$-elements in $\hat\bfx$, and the solution to the optimisation \eqref{opt:x} is thus given by
$$
\hat x_0  ~=~ \underset{x_0 \in \rea^{n}}{\arg\min} ~ \sum_{k =0}^{N-1}
 \gamma_k \Big| \bar Y(t_k) - C(t_k)\Omega(t_k) x_0 \Big|^2
$$
together with \eqref{cs:2}, and note that $\Omega$ and $\xi$ are available signals (i.e. the dynamic extension variables in the PEBO). Obviously, this exactly coincides with the solution $\hat \bfx_{\tt PEBO}$ for the case with zero initial condition $\xi_0$ for the dynamic extension. We complete the proof for the term c).
\end{proof}

The above result establishes the connection between preintegration and PEBO for the LTV dynamical model \eqref{LTV:1} with the ideal measurements of $u$. 


\section{IMU Preintegration and PEBO on Manifolds}
\label{sec3}

In this section, we extend the results in Section \ref{sec2} to the extended pose estimation problem on the manifold $SO(3) \times \rea^n$, which was the original motivation to study the preintegration approach.

\subsection{IMU Preintegration}

Let us recall the approach of IMU preintegration, which was proposed in \cite{LUPSUK} and elaborated on the manifold in \cite{FORetal}.

The motion of rigid body can be charaterised by the kinematic model
\begin{equation}
\label{kinematic}
\begin{aligned}
   \dot R &~=~ R \omega_\times
   \\
   \I \dot v & ~=~ \I a  +g
  \\
   \I \dot p & ~=~ \I v
\end{aligned}
\end{equation}
with the attitude $R\in SO(3)$, the sensor velocity $\I v \in \rea^3$, the ``apparent'' acceleration $\I a \in \rea^3$ in the inertial frame $\{I\}$, and the rigid-body position $\I p \in \rea^3$, which is briefly written as $p$. The gravity vector is given by $g = [0,0,9.8]^\top$ m/s$^2$. See \cite{BROetal} for a concise representation using the matrix group $SE_2(3)$. The IMU provides discrete-time samples of the biased acceleration and rotational velocity in the body-fixed frame $\{B\}$, i.e.,
$$
\begin{aligned}
   \B \bar a & ~=~ \B a + b_a + \epsilon_a
   \\
   \B \bar \omega & ~=~ \B \omega + b_\omega + \epsilon_\omega,
\end{aligned}
$$
in which $b_a$ and $b_\omega$ represent the sensor biases\footnote{They are slowly time-varying, but can be modelled as constants.}, and $\epsilon_a$ and $\epsilon_\omega $ are measurement noise.

\subsubsection{Standard inertial integration}
If the ``initial'' condition at $t_1$ is given, then the states $(R,\I v, \I q)$ can be uniquely obtained (for the noise-free case) as the Picard integral
{\small
\begin{align}
   R(t_2) & = R(t_1) + \int_{t_1}^{t_2} R(s) \big[\B \bar\omega(s)- b_\omega\big]_\times ds \nonumber
   \\
   \I v(t_2) & = \I v(t_1) + \int_{t_1}^{t_2}  R(s) \big( \B\bar a(s) - b_a \big)ds + \Delta_t g \label{int:1}
   \\
    p(t_2) & = p(t_1) + \Delta_t \I v (t_1)  + {1\over 2}  \Delta_t^2 g  + \iint_{t_1}^{t_2}  R(s) \big( \B\bar a(s) - b_a \big) d s^2 \nonumber
\end{align}}
\\ \
with 
$$\Delta_t:= t_2 - t_1.
$$
If $\Delta_t$ is \emph{sufficiently small}, then the first integral equation in \eqref{int:1} can be approximated by \cite{FORetal}
$$
R(t_1) \approx R(t_1) \Exp \left( \int_{t_1}^{t_2} (\B\bar\omega(s) - b_\omega) ds \right).
$$
Note that for a relatively large $\Delta t$, this does not hold.

As is shown in \cite[Sec. II]{LUPSUK}, the above standard inertial integration equations have strong nonlinearity and non-convexity with respect to the unknown initial conditions, mainly stemmed from the attitude state $R$. Between any two key frames, it requires to repeat the above ``standard'' integration, which yields heavy computational burden for real-time implementation.

\subsubsection{Inertial preintegration} 
It is well known that IMUs are sampled with a much higher rates than other sensors for navigation or localisation. In \cite{LUPSUK}, it is suggested to integrate the inertial observation \emph{between required poses in the body-fixed frame of the previous pose}, and then we may view the inertial observations as a single observation in the filter. To be precise, we may define rotation matrix $\Delta R_{t_1}^t$ related to the attitude at $t_1$, i.e.
$$
R(t) =  R(t_1)  \Delta R_{t_1}^t,
$$
with the state at $t_1$ being $\Delta R_{t_1}^{t_1} = I_3$. In general, the function $\Delta R_{t_1}^t$ does not have an analytic form, but the relative rotation matrix $\Delta R_{t_1}^{t} $ can be approximated by
\begin{equation}
\label{C:approx}
\Delta R_{t_1}^t \approx \Exp \left( \int_{t_1}^{t} \big( \B \bar\omega(s) - b_\omega \big)ds \right)
\end{equation}
for $|t-t_1| $ sufficiently small. The inertial integration \eqref{int:1} can be equivalently written as 
\begin{align}
   R(t_{k+1}) & ~=~ R(t_k)  \Delta R_{t_k}^{t_{k+1}} \nonumber
   \\
    \I v(t_{k+1}) & ~=~ \I v(t_k) + R(t_k) \Delta v_{t_k}^{t_{k+1}} + \Delta_t g \label{int:2}
   \\
    p(t_{k+1}) & ~=~ p(t_k) + \Delta_t \I v (t_k)  + {1\over 2}  \Delta_t^2 g
 + R(t_k)\Delta p_{t_k}^{t_{k+1}} \nonumber
\end{align}
with the functions for $t\ge t_k$
\begin{equation}
\label{delta:v,p}
\begin{aligned}
\Delta v_{t_k}^{t} & = \int_{t_k}^{t}  \Delta R_{t_k}^s \big( \B\bar a(s) - b_a \big) ds
\\
\Delta p_{t_k}^{t} & = \iint_{t_k}^{t} \Delta R_{t_k}^s\big( \B\bar a(s) - b_a \big) d s^2.
\end{aligned}
\quad \mbox{(IMU preintegration)}
\end{equation}
Note that the terms $\Delta v_{t_k}^{t_{k+1}}$ and $\Delta p_{t_k}^{t_{k+1}}$ are defined in the body-fixed frame, which can be calculated perfectly -- by \emph{preintegrating} IMU measurements -- without the access to the initial conditions $(R_{t_1}, \I v_{t_1}, \I p_{t_1})$. This is the original motivation to study IMU preintegration.

\subsubsection{Estimation via IMU preintegration}

The IMU preintegration has been widely used in many robotic applications, e.g., visual inertial SLAM and navigation. In these problems, there are numerous feature points, whose coordinates $p_i \in \rea^3$ ($i=1,\ldots, n_p$) are constant and unknown, i.e.,
$$
\dot p_i  =0 ,\quad i=1,\ldots, n_p.
$$
Each feature is captured by the camera, thus satisfying some algebraic equations
\begin{equation}
\label{y:NL}
 y ~=~h(x) + \epsilon_y 
\end{equation}
with $y = \rea^{n_y}$ and the noise $\epsilon_y$, which is the output function (a.k.a. observation models) in the observer theory. We have defined the extended state variable as\footnote{We assume that sensors have been well calibrated to simplify the presentation. In more general cases, we may take all biases into the variable $x$ and estimate them on-line simultaneously.}
$$
x = (R, \I v, \I p, p_1, \ldots, p_{n_p}) \in \calx
$$
with the manifold $\calx:= SO(3) \times \rea^{3(2+n_p)}$. 
At the instance $t_N$, we would like to estimate the state 
$$
\bfx (t_N):= ( x(t_0),x(t_1), \ldots,  x(t_N)).
$$
Similar to the case in Euclidean space, we may formulate it as the batch optimisation to estimate the state
\begin{equation}
\label{opt:2-1}
\hat \bfx  = \underset{\bfx \in \calx^{N} }{\arg\min} ~J_{\tt I} (\bfx)
\end{equation}
with
\begin{equation}
\label{opt:2-2}
\begin{aligned}
J_{\tt I}  ~:=~ & 
\sum_{k =0}^{N-1} \bigg[ \calj(k)
+ 
\calj_{\tt R}(k) + \calj_{\tt v}(k) + \calj_{\tt p}(k)
 \bigg]
\end{aligned}
\end{equation}
and
\begin{align}
\label{caljk}
\calj(k)  &    ~=~\big| y(t_k) - h(x(t_k)) \big|^2_{\Sigma_y^{-1}(k)} 
\\ \nonumber
\calj_{\tt R}(k) & ~=~ \Big|R(t_{k+1}) - R(t_k)\Delta R_{t_k}^{t_{k+1}} \Big|_{\Sigma_1^{-1}(k)}^2
\\ \nonumber
\calj_{\tt v}(k) & ~=~  \Big|  \I v (t_k) + R (t_k) \Delta v_{t_k}^{t_{k+1}} + \Delta_t g - \I v (t_{k+1}) \Big|_{\Sigma_{2}^{-1}(k)}^2
\\ \nonumber
\calj_{\tt p}(k) & ~=~ \Big|  p_{t_k} + \Delta_t \I v _{t_k} + {1\over 2}  \Delta_t^2 g
 + R_{t_k}\Delta p_{t_k}^{t_{k+1}} - p_{t_{k+1}} \Big|_{\Sigma_3^{-1}}^2
\end{align}
and $\Sigma_i \succ 0$ ($i=1,2,3$) are some covariances to characterise the uncertainty in the model \eqref{int:2}. If the stochastic properties of $\epsilon_a$ and $\epsilon_w$ are known in advance, we may use some on-line propagation to \emph{approximate} $\Sigma_i(k)$. See \cite{FORetal,BROetal} for example, and we omit its details.

\vspace{0.2cm}

\fbox{ \parbox { .9\linewidth} 
{
{\em Estimation via IMU Preintegration on Manifolds:} 
\begin{itemize}
    \item[-] preintegration: \eqref{C:approx}, \eqref{delta:v,p}
    \item[-] estimate: $\hat \bfx_{\tt PI}$
    \item[-] optimisation: \eqref{opt:2-1}-\eqref{opt:2-2}
\end{itemize}
}}
\vspace{0.2cm}

\subsection{Parameter Estimation-Based Observer on Manifolds}

In this section, we briefly summarise the main results in our previous papers \cite{YIetalAUT,YIetalCDC,YIWAN} about the PEBO design on manifolds. 

Consider the kinematics \eqref{kinematic} with the measurable output in \eqref{y:NL}. In \cite{YIetalAUT}, the observer design is conducted in the body-fixed frame with the dynamics given by
$$
\begin{aligned}
\dot R & ~=~ R\omega_\times 
\\
\B \dot v & ~=~ -\omega_\times \B v + \B \bar a - b_a + R^\top g
\\
\B \dot p & ~=~ -\omega_\times \B p - \B v,
\end{aligned}
$$
where $\B p$ is defined as the origin coordinate of $\{I\}$ in the body-fixed frame, i.e. 
$$
\B p := R^\top \I p.
$$
In the PEBO approach, we design the dynamic extension
\begin{equation}
\label{dyn_ext:m}
\begin{aligned}
 \dot Q & ~=~ Q \omega_\times
 \\
 \dot \xi & ~=~  A(\omega, Q)\xi + B(\B \bar a , b_a)
 \\
 \dot \Omega & ~=~ A(\omega, Q) \Omega
 \\
 \Omega(t_0) & ~=~ I,
\end{aligned}
\end{equation}
with 
$$
\begin{aligned}
A(\omega,Q ) &:= \begmat{-\omega_\times & 0 & Q^\top \\ -I & -\omega_\times & 0 \\ 0 & 0 & 0},
\\
B(\B \bar a, b_a) & := \begmat{\B \bar a - b_a \\ 0 \\ 0}.
\end{aligned}
$$
The key observation in \cite{YIetalAUT} is that the system state can be linearly parameterised as
\begin{equation}
\label{state:2}
\begin{aligned}
R_t & ~=~ Q_c Q^\top_t
\\
\begmat{\B v\\ \B p\\ ~g_c} & ~=~ \xi_t - \Omega_t\xi_0 + \Omega_t \theta
\end{aligned}
\end{equation}
with the unknown constant matrix $Q_c \in SO(3)$, and the vector 
$$
\theta:= \col(\B v(0), \B p(0), g_c).
$$
Similar to the case in Euclidean space, we only need to determine $(Q_c,\theta)$ and $ \bfp:=(p_1, \ldots, p_{n_p})$, whose estimates are written as $(\hat Q_c, \hat\theta,\hat \bfp)$. Then, the estimates of $x\in \mathcal{X}$ is given by
\begin{equation}
\label{pebo:y1}
\hat x_t  = (\hat R, \hat R \B \hat v ,\hat R \B \hat p , \hat \bfp).
\end{equation}
with
\begin{equation}
\label{pebo:y2}
\begin{aligned}
\hat R & ~= ~ \hat Q_c Q_t
\\
\begmat{\B \hat v, \B \hat p, \hat g_c}^\top & ~=~ \xi_t - \Omega_t \xi_0 + \Omega_t \hat \theta.
\end{aligned}
\end{equation}

For the measurements collected at instances $\{t_k\}$, the unknown $(\hat Q_c, \hat \theta,\hat \bfp)$ can be obtained from the following optimisation:
\begin{equation}
\label{opt:pebo_manifold}
\begin{aligned}
& (\hat Q_c, \hat \theta, \hat \bfp) ~=~ \underset{\substack{Q_c \in SO(3) \\  \theta \in \rea^9, \bfp \in \rea^{3n_p}} }{\arg\min} ~ \sum_{k=0}^{N-1} \calj(k)
\\
& \mbox{s.t.} \qquad \hat g_c~=~ \hat Q_c^\top g
\end{aligned}
\end{equation}
with $\calj$ defined in \eqref{caljk}. The main result of PEBO on manifolds is summarised as follows.

\vspace{0.2cm}

\fbox{ \parbox { .92\linewidth} 
{
{\em PEBO on manifolds:} 
\begin{itemize}
    \item[-] dynamics: \eqref{dyn_ext:m}
    \item[-] estimate (observer output): $\hat \bfx_{\tt PEBO}$ from \eqref{pebo:y1}-\eqref{pebo:y2}
    \item[-] optimisation: \eqref{opt:pebo_manifold}
\end{itemize}
}}
\vspace{0.2cm}

\subsection{The PEBO Viewpoint to IMU Preintegration}

We are in the position to present the main result of the paper. Similarly to the case in Euclidean space, we establish the connection between IMU preintegration and PEBO on manifolds as follows.

\begin{proposition}
\label{prop:so3}\rm
Consider the kinematics \eqref{kinematic} with constant $p_i$ ($i=1,\ldots, n_p$). The estimation of the state of the IMU preintegration \eqref{C:approx}-\eqref{opt:2-2} converges to the estimate of the PEBO \eqref{dyn_ext:m}-\eqref{opt:pebo_manifold} as $\min_{j=1,2,3}(\lambda_{\tt max}\{\Sigma_j\}) \to 0$, in the following sense.
\begin{itemize}
    \item[a)] The preintegration of $\Delta R_s^t$ and the extended state $Q$ satisfy
    \begin{equation}
    \begin{aligned}
        Q(t_0)^\top Q(t_k)  ~=~ \prod_{i=0}^{k-1}{'} \Delta R_{t_k}^{t_{k+1}} := \Delta R_{t_0}^{t_1} \ldots \Delta R_{t_{k-1}}^{t_k}
    \end{aligned}
    \end{equation}
    for all  $k\in \mathbb{N}$.
    
    \item[b)] If the cost function $\calj$ has a global minimum, then the estimates from the PEBO and the IMU preintegration satisfy 
    $$
    \hat\bfx_{\tt PI} \to \hat\bfx_{\tt PEBO} \quad \mbox{as} \quad \lambda_{\tt max}\{\Sigma_j\} \to 0 ~(j=1,2,3).
    $$
\end{itemize}
\end{proposition}
\begin{proof}
The property a) is straightforward to verify because 
$$ \Delta R_{t_0}^{t_1} \ldots \Delta R_{t_{k-1}}^{t_k}  = \Delta R_{t_0}^{t_k} $$ and 
$$
{d\over dt}(RQ^\top) = 0.
$$
When the largest eigenvalue of $\Sigma_j$ converges to zero, the last three terms in \eqref{opt:2-2} make \eqref{int:2} as the hard constraints. For the fact b), we need to show that the constraint \eqref{pebo:y2} together with 
$$
\hat g_c = \hat Q_c^\top g
$$ 
in PEBO yields the constraint \eqref{int:2} in IMU preintegration. To see this, for a fixed (constant) estimate $\hat\theta$ and defining 
$$
\eta := \col(\B \hat v, \B \hat p, \hat g_c)
$$ 
we have
$$
\begin{aligned}
\dot \eta
& ~=~ \dot \xi - \dot \Omega \xi_0 + \dot \Omega \hat\theta
\\
& ~=~ A(\omega, Q) \xi + B(\B \bar a,b_a) - A(\omega, Q) \Omega(\xi_0 +  \hat\theta)
\\
& ~=~ A(\omega, Q) \eta + B(\B \bar a, b_a).
\end{aligned}
$$
Now, consider the coordinate transformation 
$$
\eta \mapsto z= \begin{pmatrix}
    z_1\\ z_2 \\z_3
\end{pmatrix}
:= \begin{pmatrix}\hat R \B\hat v\\ \hat R\B \hat p\\ \hat Q_c \hat g_c \end{pmatrix} .
$$
In the transformed coordinate, the dynamics verifies 
$$
\begin{aligned}
\dot z_1 & = R(\B a- b_a) + g 
\\
\dot z_2 & = z_1
\\
\quad \dot z_3 & = 0.
\end{aligned}
$$
Considering the constraint in \eqref{opt:pebo_manifold}, we may equivalently select the decision variable as $(\hat R, z_1,z_2, \hat \bfp)$, and the change of decision variable does not affect the minimum of the cost function $\calj$.

In the new coordinate, $z_1$ and $z_2$ satisfy
\begin{align}
    z_1(t_{k+1}) & ~=~ z_1(t_k) + R(t_k) \Delta v_{t_k}^{t_{k+1}} + \Delta_t g \nonumber
   \\
    z_2(t_{k+1}) & ~=~ z_2(t_k) + \Delta_t \I v (t_k)  + {1\over 2}  \Delta_t^2 g
 + R(t_k)\Delta p_{t_k}^{t_{k+1}}. \nonumber
\end{align}
It exactly coincides with \eqref{int:2}. Hence, following the same arguments in the proof of Proposition \ref{prop:eul}, we can show that  the estimates from these two approaches are exactly the same.
\end{proof}

\section{Discussion and Applications}
\label{sec4}
\subsection{Discussions}

In this section, we present some further remarks and applications following from the connections between pre-integration and PEBO.

\begin{remark}\rm 
First, let us make some comparisons between two frameworks of PEBO and preintegration.
\begin{itemize}
    \item[\labelitemii] The preintegration approach may be roughly viewed as the implementation of PEBOs in a moving horizon, i.e., the ``initial moment'' is recursively defined as $\{t_k\}_{k\in \mathbb{N}}$ and then the task is to estimate the state $x(t_k)$. In PEBO, we only need to estimate the initial condition at $t_0$. For the ideal case with perfect models and measurements, these two frameworks exactly coincide with each other, as illustrated in Proposition \ref{prop:eul}.
    \item[\labelitemii] In the pose estimation-related problems, the IMU preintegration utilises the body-fixed frame for accelerations and velocities; in contrast, the PEBO in our previous works \cite{YIetalAUT,YIetalCDC} adopts the inertial frame. 
\end{itemize}
\end{remark}

\begin{remark}\rm
In IMU preintegration, it is possible to write the state transition matrix analytically for the $(\I v, \I p)$-subsystem; see \eqref{int:2}. In PEBO, we need to calculate the state transition matrix for the $(\B v, \B p)$-subsystem numerically, but it brings two benefits:
    \begin{itemize}
        \item[B1:] The sensor bias $b_a$ appears in the dynamics \eqref{dyn_ext:m} in a linear way. As shown in \cite{YIetalAUT}, we are able to construct a linear regression model on the unknown bias $b_a$ using the PEBO methodology.
        \item[B2:] In some applications, we do not need the estimation of attitude $R$. By applying PEBO in the body-fixed frame, we are able to estimate $(\B v, \B p, \bfp)$ directly without the information of attitude.
    \end{itemize}
\end{remark}


\begin{remark}\rm 
In the generalised PEBO approach \cite{ORTetalAUT}, there is a need to calculate the fundamental matrix $\Omega(t)$ over time in \eqref{dot:omega}. Though its dynamics is forward complete, the variable $\Omega$ is unbounded when the matrix $A_t$ is unstable. Since $\Omega$ is part of the internal state in the observer, at some finite time the observer would become dramatically ill-conditioned and impossible to represent accurately in memory. As a result, it may bring some numerical issues and make the observer very sensitive to sorts of perturbations. For this consideration, it is reasonable to implement a PEBO in ``moving horizons'' like preintegration in order to improve robustness.
\end{remark}

\begin{remark}\rm
When considering the uncertainty from the input-output measurements, the estimates from the PEBO and preintegration approaches would be different. In PEBO, we only need to solve the optimisation problem with the decision variable $\theta$ (equivalently $x_0$) at a single instance; in contrast, the hard constraint \eqref{cs:1} does not hold in the preintegration approach, and there are additional decision variables $\{x_k, w_k\}_{k\in \mathbb{N}}$. For this case, their relation resembles the \emph{single} and \emph{multiple shootings} in the direct methods for \emph{optimal control}.
\end{remark}


\begin{remark}\rm
State estimation via recursive algorithms under Assumption \ref{ass:1} is known as the problem of sampled-data (or digital) observers \cite{MAN,ARCDRA}. Even for linear time-invariant (LTI) systems, there are still several open problems to design a sampled-data observer \cite{SFEetal}. An useful application of the proposed equivalence between preintegration and PEBO is providing a novel method to design sampled-data observers. We will present constructive details in the next subsection.
\end{remark}

\subsection{Application I: Sampled-data Observer via Preintegration}

In this section, we show that the proposed equivalence provides a new method to design a hybrid sampled-data observer for the LTV system \eqref{LTV:1}. We summarise the results as follows. To simplify the presentation, as well as to obtain asymptotic stability claims, we consider the ideal measurements $(u,y)$ in the following proposition.

\begin{proposition}\label{prop:4}\rm
Consider an observable LTV system \eqref{LTV:1}. Assume the sampled instances $\{t_k\}_{k \in \mathbb{N}}$ are selected such that
\begin{itemize}
\item[P1:] The pair $(\Phi(t_{k+1}, t_{k}), C(t_k))$ is (discrete-time) uniformly completely observable, where $\Phi(\cdot,\cdot)$ is the continuous-time state transition matrix of $A_t$ defined in \eqref{dot:omega}.
\item[P2:] There exists a constant $k_2 \in \mathbb{N}_+$ such that
\begin{equation}
\label{Wq}
  W_q := \sum_{i= k}^{k+k_2} \Psi(i,k) Q \Psi^\top(i,k)  \succ \delta_q I_n  
\end{equation}
for some $Q \succ 0, \delta_q>0$ and $\forall k\in \mathbb{N}$ with $\Psi(i,k)$ the discrete-time state transition matrix of $z_{k+1} = \Phi(t_{k+1},t_{k}) z_k$.
\end{itemize}
Then, the hybrid sampled-data observer
\begin{equation}
\begin{aligned}
& ~~~\left.
\begin{aligned}
   \dot F & = A_t F, \quad \hspace{0.9cm} F(t_k^+) = I_n
   \\
   \dot v & = A_t v + B_tu, \quad  v(t_k^+) = 0_n.
   \\
   F_k & = F(t_{k+1}^-), \quad \hspace{0.4cm} v_k \hspace{0.45cm} = v(t_{k+1}^-).
\end{aligned}
~ \hspace{1.4cm}\right\}
~  \mathcal{H}_1
\\ 
&
\left.
\begin{aligned}
\hat x_{k+1}  & = F_k \hat x_k + v_k + K_{k+1} e_{k+1}
\\
e_{k+1} & =  y_{t_{k+1}} - C_{t_{k+1}} (F_k \hat x_k + v_k) - D_{t_{k+1}}u_{t_{k+1}}
\\
K_k & = \hat P_k C_k^\top [C_k \hat P_k C_k^\top + R]
\\
\hat P_{k+1} & = F_k P_k F_k^\top + Q 
\\
P_k & = \hat P_k - K_k C_k \hat P_k.
\end{aligned}
\hspace{0.1cm}
\right\}
~  \mathcal{H}_2
\end{aligned}
\end{equation}
with some positive definite matrices $Q$ and $R$, provides a globally asymptotically convergent estimate $\hat x$, i.e.
\begin{equation}
\label{convergence:hybrid_observer}
\lim_{k\to \infty} |\hat x_k - x(t_k)| =0. 
\end{equation}
\end{proposition}
\begin{proof}
According to Propositions \ref{prop:1}-\ref{prop:eul}, the systems state $x$ at the instances $\{x(t_k)\}_{t\in\mathbb{N}}$ exactly satisfies the discrete dynamical model
\begin{equation}
\label{DT_model}
\begin{aligned}
   x(t_{k+1}) & ~=~ F_k  x(t_k) + v_k 
   \\
   y(t_{k}) & ~=~ C(t_k) x(t_{k}) + D(t_k) u(t_k),
\end{aligned}
\end{equation}
with the preintegration signals $F_k$ and $v_k$ generated from the system $\calh_1$. Invoking the first equation in \eqref{itema}, we have 
$$
F_k = \Omega(t_{k+1}) \Omega^{-1}(t_k) = \Phi(t_{k+1}, t_k).
$$
As a consequence, the discrete-time uniform complete observablility (UCO) of the pair $(\Phi(t_{k+1},t_k), C(t_k))$ implies the UCO of the LTV system \eqref{DT_model}. Note that the system $\calh_2$ is the standard Kalman-Bucy filter for the LTV system \eqref{DT_model}. Together with the condition \eqref{Wq}, we conclude the global asymptotic convergence \eqref{convergence:hybrid_observer} by invoking \cite[Thm. 4.1]{AND}.
\end{proof}

\begin{remark}\rm 
In the condition P1, it is equivalent to impose the UCO of the discrete-time LTV system \eqref{DT_model}. It is relatively straightforward to verify the UCO of the continuous-time system \eqref{LTV:1} is a necessary condition to P1, but it is not sufficient. Consider the constant observable pair $(A_0, C_0)$, and let $A = A_0$, $C(t) = C_0$ for $t\in [2k, 2k+1)$ and $C(t)= 0$ for $t\in [2k+1, 2k+2)$ with $k \in \mathbb{N}$. The resulting pair $(A_t,C_t)$ guarantees the UCO of \eqref{LTV:1} but not for the system \eqref{DT_model} if the sampled data are collected in $[2k+1, 2k+2)$. On the other hand, the condition P1 is unnecessary to design a sampled-data observer. If the observability Gramian is positive definite only in some interval but not uniform over time, it is still possible to design globally convergent state observer by using MHE or some state-of-the-art recursive designs \cite{WANetal,YIetalAUT,YIetalCDC}.
\end{remark}

\begin{remark}\rm 
In \cite{ARCDRA}, nonlinear sampled-data observers are classified into two categories: i) design via approximate discrete-time models of the plant; and ii) emulation: discretisation of continuous-time observers. Clearly, the proposed observer belongs to the first class, but we utilise an \emph{exact} discrete-time model rather than its approximation because of its linearity. Indeed, the proposed design is also applicable to nonlinear systems which can be transformed into the affine form.
\end{remark}

\begin{remark}\rm
The proof of Proposition \ref{prop:4} does \emph{not} rely on the assumption of \emph{periodic} sampling. That is, the proposed sampled-data observer is also immediately applicable to the case with \emph{asynchronous} measurements, which was studied for the linear time-invariant (LTI) systems \cite{SFEetal}. We provide a much simpler solution to this specific problem for LTV systems.
\end{remark}

\subsection{Application II: Statistical Optimality in PEBO}
\label{sec:43}

In this subsection, we will show that the proposed equivalence in Sections \ref{sec3}-\ref{sec4} leads to an intuitive way to improve the performance of PEBO in the presence of noisy input $u$. 

We assume that the initial condition $x_0$ is a \emph{deterministic} variable but unknown, and model the noisy terms $\epsilon_u$ and $\epsilon_y$ as zero-mean white noise processes \eqref{E:sigma_uy}, in which $\epsilon_u \in \rea^m $ and $ \epsilon_y\in \rea^p$ are addictive zero-mean white-noise processes, namely\footnote{Here, the white processes are not rigorously defined due to the $\delta$-covariances, with $\delta$ the delta function. A rigorous definition is based on the stochastic differential equations \cite[Sec. 7]{ASTBook}.}
\begin{equation}
\label{E:sigma_uy}
\begin{aligned}
    \mathbb{E}[\epsilon_{u,t} \epsilon_{u,s}^\top] & = \Sigma_u \delta(t-s)\\
    \mathbb{E}[\epsilon_{y,t} \epsilon_{y,s}^\top] & = \Sigma_y \delta(t-s).
\end{aligned}
\end{equation}
The variables $\{\epsilon_u\}, \{\epsilon_y\}$ and $x_0$ are uncorrelated. Then, the error $e= x-\xi$ in PEBO for the LTV system \eqref{LTV:1} satisfies
\begin{equation}
\label{e:noisy}
\dot e = A_t e + B_t \epsilon_u.
\end{equation}
According to the state covariance propagation for LTV systems \cite[Ch. 4]{davis1977linear}, we have
$$
x_t - \xi_t = \Omega_t(\theta - \xi_0) + \epsilon_e
$$
with the white-noise process $\epsilon_e$, i.e.
$
\mathbb{E}[\epsilon_{e}(t) \epsilon_{e}(s)^\top] = \Pi_t \delta(t-s)
$
and $\Pi_t$ satisfies
\begin{equation}
\label{dot:Pi}
    \dot \Pi_t = A_t \Pi_t + \Pi_t A_t^\top + B_t \Sigma_u B_t^\top, \quad \Pi(0)= 0_{n\times n},
\end{equation}
where the initial condition of $\Pi$ is due to the deterministic assumption of $x_0$. Noting that the uncertainties from $\bar u$ and $\bar y$ in \eqref{Y}, we have
\begin{equation}
    \bar Y = C_t\Omega_t \theta + \epsilon_{\tt Y}
\end{equation}
with 
$$
\epsilon_{\tt Y}(t) := \epsilon_y - D_t\epsilon_u + C_t\epsilon_e.
$$

Unfortunately, the variables $\epsilon_e$ and $\epsilon_u$ are \emph{not} independent, since $\epsilon_e(t)$ is indeed filtered from $\epsilon_u$. However, for the LTV system \eqref{LTV:1} without the feedfoward term, i.e. $D_t=0$, the variable $\epsilon_{\tt Y}$ is a white noise process
$$
\mathbb{E}[\epsilon_{\tt Y}(t)\epsilon_{\tt Y}(s)^\top] = \left(\Sigma_y + C_t\Pi_t C_t^\top\right) \delta(t-s).
$$
Hence, we may reformulate the optimisation \eqref{hat:theta} as
\begin{equation}
\hat\theta:= \underset{\theta \in \rea^n}{\arg\min} ~ \sum_{k=0}^{N-1} \Big| \bar Y(t_k) - C(t_k)\Omega(t_k)\theta \Big|^2_{(\Sigma_y + C_t\Pi_tC_t^\top)^{-1}}
\end{equation}
to obtain some statistic optimality, where $\Pi_t$ is generated from \eqref{dot:Pi}.

\begin{remark}\label{rem:stochastic}
\rm
In \cite{ORTetalAUT}, the PEBO approach is applicable to nonlinear systems in the form of
$$
\dot x = f(x,u), \quad y= h(x,u),
$$
for which a coordinate transformation $x\mapsto z:=\phi(x)$ exists such that the lifted dynamics is given by
\begin{equation}
\label{dot:z}
    \dot z = A(u,y)z + B(u,y) , \quad y = C(u,y)z + D(u,y).
\end{equation}
It is generally difficult to calculate covariance propagation for nonlinear systems, but there are many works discussing how to empirically \emph{approximate} it in the literature on preintegration \cite{BROetal,FORetal,LUPSUK}. The proposed connection between two approaches, together with the state-of-the-art development of preintegration, provides a promising way to develop nonlinear stochastic PEBO method.
\end{remark}

\section{Concluding Remarks}
\label{sec5}

In this paper, we have presented a novel observer interpretation to the IMU preintegration approach. Our findings reveal an exact correspondence between the preintegrated signals and the dynamic extended variables in PEBO that is implemented in a moving horizon. Furthermore, we have identified the precise conditions under which these two approaches yield identical estimates. These results were developed in both the Euclidean space and matrix Lie groups. Finally, we have utilised the proposed equivalence to design a novel sampled-data observer for LTV systems, and to improve the performance of PEBO in the presence of measurement noise.

These connections suggest some interesting avenues for future research, including: 

\begin{itemize}
    
    \item[-] In the preintegration and PEBO approaches, we require that the system dynamics is in (or can be transformed into) a state-affine form \eqref{dot:z}. It would be interesting to integrate them with contraction analysis \cite{LOHSLO,BUL}, for which the so-called differential dynamics is exactly an LTV system. 
    
    \item[-] In Section \ref{sec3}, we show that different coordinates are used in the IMU preintegration and PEBO. For the latter, we adopt the body-fixed coordinate $(\B v, \B p)$, and it is interesting to observe the benefit of the linear parameterisation of bias $b_a$. This is notable by its absence in the inertial coordinate for preintegration \cite{FORetal}. Hence, it would be of practical interest to implement IMU preintegration in the body-fixed coordinate towards real-time bias estimation. 

    \item[-] It is theoretically interesting to elaborate the results in Section \ref{sec:43} using It\^{o} integrals toward a more rigorous formulation.
\end{itemize}

\appendix
\section*{Appendix}
\section{Some Definitions}

\begin{definition}
\rm 
A pair $(A_k,C_k)$ of discrete-time systems is uniformly completely observable if the observability Gramian 
$$
W_O[k,k_1] \succeq \delta_o I
$$
for some $\delta_o>0$, $k_1 \in \mathbb{N}_+$ and \emph{all} $k \ge 0$, with 
\begin{equation}
W_O[k,k_1]:= \sum_{i=k}^{k+ k_1} \Psi^\top(i,k) C_k^\top C_k \Psi(i,k)    
\end{equation}
in which $\Psi(i,k)$ is the state transition matrix from $k$ to $i$ of the system $z_{k+1} = A_k z_k$. 
\end{definition}

\section*{CRediT authorship contribution statement}

\textbf{Bowen Yi:} Conceptualization, Methodology (propositions), Writing - original draft. \textbf{Ian R. Manchester:}  Methodology, Writing - review and edit, Project administration.

\section*{Declaration of competing interest}

The authors declare that they have no known competing financial interests or personal relationships that could have appeared to influence the work reported in this paper.

\section*{Acknowledgement}

This paper is supported by the Australian Research Council. The first author would like to thank Dr. Chi Jin for bringing IMU preintegration into his attention.

\bibliographystyle{abbrv}
\bibliography{reference}

\end{document}